\documentclass[10pt]{IEEEtran}
\usepackage{etex}

\usepackage{algpseudocode}
\usepackage{algorithm}
\usepackage{amsmath}
\usepackage[utf8]{inputenc}
\usepackage{amsthm}
\usepackage{amssymb}

\usepackage{varwidth}
\DeclareMathOperator*{\argmax}{arg\,max}
\usepackage[noadjust]{cite}
\newcommand{\blue}[1]{\textcolor{black}{#1}}
\usepackage{enumerate}
\usepackage{amssymb}
\usepackage{cite}
\newcommand{\vc}[1]{{\mathbf{#1}}}
\usepackage{multirow}
\usepackage{dsfont}
\usepackage{array}
\newtheorem*{definition}{Definition}
\newtheorem{lemma}{Lemma}
\newtheorem{theorem}{Theorem}
\usepackage{cite}
\usepackage{citesort}

\usepackage{listings}
\newcolumntype{L}[1]{>{\raggedright\let\newline\\\arraybackslash\hspace{0pt}}m{#1}}
\newcolumntype{C}[1]{>{\centering\let\newline\\\arraybackslash\hspace{0pt}}m{#1}}
\newcolumntype{R}[1]{>{\raggedleft\let\newline\\\arraybackslash\hspace{0pt}}m{#1}}
\usepackage{graphicx}
\usepackage{float}
\usepackage{epstopdf}
\usepackage{bbm}
\usepackage{pbox}
\usepackage{color}
\usepackage[section,subsection,subsubsection]{extraplaceins}
\usepackage{multirow}
\usepackage{subfigure}
\usepackage{booktabs}

\newcommand{\RNum}[1]{\uppercase\expandafter{\romannumeral #1\relax}}
\newcommand{\suchthat}{\;\ifnum\currentgrouptype=16 \middle\fi|\;}

\usepackage{amsmath}
\usepackage{tikz}
\usetikzlibrary{automata,arrows,positioning,calc}

\usetikzlibrary{arrows}
\tikzstyle{int}=[draw, fill=blue!20, minimum size=2em]
\tikzstyle{init} = [pin edge={to-,thin,black}]

\usetikzlibrary{positioning,chains,fit,shapes,calc, arrows, shadows,trees}
\usepackage{calc}
\usepackage{multicol}
\usepackage{lipsum}
\usetikzlibrary{decorations.markings, arrows, automata}
\tikzstyle{vertex}=[circle, draw, inner sep=0pt, minimum size=5pt]
\tikzset{symbol/.style={rectangle, draw, very thick,
minimum size=10mm, rounded corners=1mm}}
\tikzset{symbol2/.style={rectangle , draw,  thick,
minimum size=35mm, rounded corners=1mm}}

\tikzstyle{block} = [draw, fill=white, rectangle, 
    minimum height=3em, minimum width=6em]
\tikzstyle{sum} = [draw, fill=gray, circle, node distance=1cm]
\tikzstyle{input} = [coordinate]
\tikzstyle{output} = [coordinate]
\tikzstyle{pinstyle} = [pin edge={to-,thick,black}]

\usepackage{listings}
\begin{document}
\title{\centering {{Efficient Beam Alignment in Millimeter Wave Systems Using Contextual Bandits}}}
\author{Morteza Hashemi, Ashutosh Sabharwal,
        C.~Emre~Koksal,
        and~Ness~B.~Shroff
\thanks{M. Hashemi and C. E. Koksal are with the Department
of Electrical and Computer Engineering, Ohio State University, Columbus,
OH, 43210 USA, e-mail: hashemi.20@osu.edu, 	koksal.2@osu.edu}
\thanks{A. Sabharwal is with the Department
of Electrical and Computer Engineering, Rice University, Houston, TX 77005, USA, e-mail: ashu@rice.edu.}
\thanks{N. B. Shroff is with the Department
of Electrical and Computer Engineering and Department of Computer Science and Engineering, Ohio State University, Columbus,
OH, 43210 USA, e-mail: shroff.11@osu.edu.}}
\maketitle
\thispagestyle{plain}
\pagestyle{plain}

\begin{abstract}
In this paper, we investigate the problem of beam alignment in millimeter wave (mmWave) systems, and design an optimal algorithm to reduce the overhead. Specifically, due to directional communications, the transmitter and receiver beams need to be aligned, which incurs high delay overhead since without a priori knowledge of the transmitter/receiver location, the search space spans the entire angular domain. This is further exacerbated under dynamic conditions (e.g., moving vehicles) where the access to the base station (access point) is highly dynamic with intermittent on-off periods, requiring more frequent beam alignment and signal training. To mitigate this issue, we consider an online stochastic optimization formulation where the goal is to maximize the directivity gain (i.e., received energy) of the beam alignment policy within a time period. We exploit the inherent correlation and unimodality properties of the model, and demonstrate that contextual information improves the performance. To this end, we propose an equivalent structured Multi-Armed Bandit  model to optimally exploit the exploration-exploitation tradeoff. In contrast to the classical MAB models, the contextual information makes the lower bound on \emph{regret} (i.e., performance loss compared with an oracle policy)  independent of the number of beams. This is a crucial property since the number of all combinations of beam patterns can be large in transceiver antenna arrays, especially in massive MIMO systems. We further provide an asymptotically optimal beam alignment algorithm, and investigate its performance via simulations.
\end{abstract}

\section{Introduction}
To meet the exponentially growing demand in mobile data, the trend in wireless networks is migrating to higher frequencies combined with increasing number of antennas per device and per base station.  For instance, it is envisioned that in $5$G cellular systems certain portions of the millimeter wave (mmWave) band will be used, spanning the spectrum between $30$GHz to $300$GHz. However, propagation loss at mmWave frequencies is much higher due to a variety of factors including atmospheric absorption, basic Friis transmission-effect, and low penetration. When the users and/or surrounding objects are mobile, this effect is more pronounced such that different propagation paths become highly variable with intermittent on-off periods. Thus, unlike existing communication schemes, mmWave systems require highly directional communications to compensate for large channel losses. Thanks to recent advances in antenna technologies, large directional antenna arrays with much smaller form factors can be deployed in relatively small chip areas. Such arrays have the potential to focus the signal energy toward a specific direction, ``making up'' for the channel losses. 

In order to fully utilize directional communications, the transmitter and receiver beams need to be aligned. The experimental results in \cite{nitsche2015steering} demonstrate that in a system with a $7$ degree beam width, a misalignment of $18$ degree reduces the link
budget by around $17$ dB, which can reduce the maximum throughput
by up to $6$ Gbps or break the link entirely \cite{IeeeCode}. On the other hand, as a result of such ``pencil-beams" at the transmitter and receiver, beam alignment incurs a large overhead that scales with device mobility and the product of the transmitter-receiver beam resolution. 
In exhaustive search methods, both users and base stations have a predefined codebook of beam directions that cover the
entire angular space and are used sequentially to transmit and receive. Thus, the complexity of this exhaustive search is $O(N^2)$, where $N$
is the number of possible beam directions. To improve the search efficiency, the transmitter and receiver steering is decoupled in the 802.11ad standard such that the transmitter starts with a quasi-omnidirectional beam, while the receiver scans the space for the best beam direction. The process is then
reversed \cite{zhou2012efficient}. This approach reduces the search complexity to $O(N)$. Still, for a beam of a few degrees, the delay can be
hundreds of milliseconds to seconds \cite{zhu2014demystifying}, which would
easily stall real-time applications.  

Dynamic conditions make the beam alignment more challenging since there is the need for frequent beam alignment. Under such scenarios, we pose the following question that \emph{given the outcome of the past beam alignments, is it possible to extract some information and reduce the search space for the subsequent beam alignment procedures?} In particular, our work is based on the fact that successive beam alignments are stochastically correlated, and thus, outcome of the previous ``beam matching'' provides \emph{contextual information} for the subsequent matchings, thus eliminating the need to search the entire angular domain. We exploit \emph{correlation} and \emph{unimodality} properties across various beam matching.  Specifically, for a given beam matching, we call the difference between the transmitter and receiver direction as \emph{misalignment}. Because of correlation if matching at a larger misalignment is successful (i.e., received energy is above a threshold $\tau$), with a high probability a matching will be successful at a smaller misalignment as well. Furthermore, the directivity gain (or received energy) can be approximated as a unimodal function of the misalignment value. We exploit this contextual information in order to obtain a beam search scheme that quickly identifies the best beam direction and maximizes the directivity gain. We formulate the problem of finding the best beam pair as an online stochastic optimization where the objective is to maximize the expected amount of received energy within a given time period. 

To find the optimal solution, we show that this problem can be considered as an instance of the Multi Armed Bandit (MAB) model in which \emph{each transmit and receive beam pair is considered as a single arm.} Thus, the objective is to design a sequential arm selection (or, equivalently, beam alignment) strategy that maximizes the \emph{expected reward} (received energy) over a given time horizon.  Performance of MAB models is usually expressed in terms of \emph{regret} that is defined as the total expected reward loss compared with an oracle policy. Regret of the best algorithm is in the form of $O(K \log(T))$ in which $K$ is the number of arms and $T$ is the time horizon \cite{lai1985asymptotically}. In distinction with the classical MAB models, we exploit the contextual information of beam alignment that leads to a structured MAB model, and prove that \emph{the regret does not scale with the number of arms that is equal to the number of beam matchings.} This is a crucial property in  (massive) MIMO systems, and provides a fundamental performance limit
satisfied by any exhaustive beam selection algorithm. This limit quantifies the inevitable performance loss due to the need to explore sub-optimal beam pairs. It also characterizes the performance gains that can
be achieved by devising beam pair selection schemes that optimally exploit the correlations and the structural properties of the MAB problem. 
\emph{Therefore, in contrast to the previous works that explore the sub-optimal beam pairs by heuristics, our method optimally explores the sub-optimal beam pairs.} The following example illustrates how we achieve this goal. 

{\emph{\textbf{Example}: Let us consider a scenario where the transmitter beam direction is fixed at $75^{\circ}$ angle with respect to the receiver. For the sake of exposition, we assume a 2D setting. Assume that there are $16$ possible directions at the receiver, as shown in Fig. \ref{fig:polar-example}. Using the exhaustive beam selection scheme, each of 16 directions will be examined one at a time, and the direction with the largest received energy (from beacon messages) is picked. However, under dynamic conditions  (e.g., with mobility), the optimal beam direction can potentially change within a short period of time. In this case, we consider maximizing the received energy within a given period of time. Using our proposed scheme, the receiver assigns an index to each beam direction, and the beam with the highest index will be selected. The important point is that due to the correlation and unimodality properties, the search space will be limited to the neighborhood of the beam with maximum index. As a result, it prevents the need of a uniform exploration over the entire angular domain, thus mitigating the overhead of beam alignment when the number of beam directions becomes large.}} 

In summary, our contributions are as follows:

\begin{itemize}
\item We consider the beam alignment problem, and investigate the fundamental performance limits of the search-based beam alignment between the transmitter and receiver antennas. We model the problem of finding the best beam alignment as an online stochastic optimization problem. 

\item We exploit contextual information of the problem and formulate an equivalent structured Multi-Armed Bandit model. We experimentally demonstrate that the received power (approximately) follows a unimodal pattern.  

\item We derive a lower bound on the regret of any search-based algorithm, and demonstrate that the regret does not scale with the number of transmission and receive beams, thanks to the underlying structure of the problem. \blue{Finally, based on the OSUB algorithm in \cite{combes2014unimodal}, we propose the Unimodal Beam Alignment (UBA) algorithm that is shown to be asymptotically optimal.}
\end{itemize}


\begin{figure}[t]
\centering
\includegraphics[scale=.5, trim = 1.6cm 2.6cm 1.5cm 2.2cm, clip]{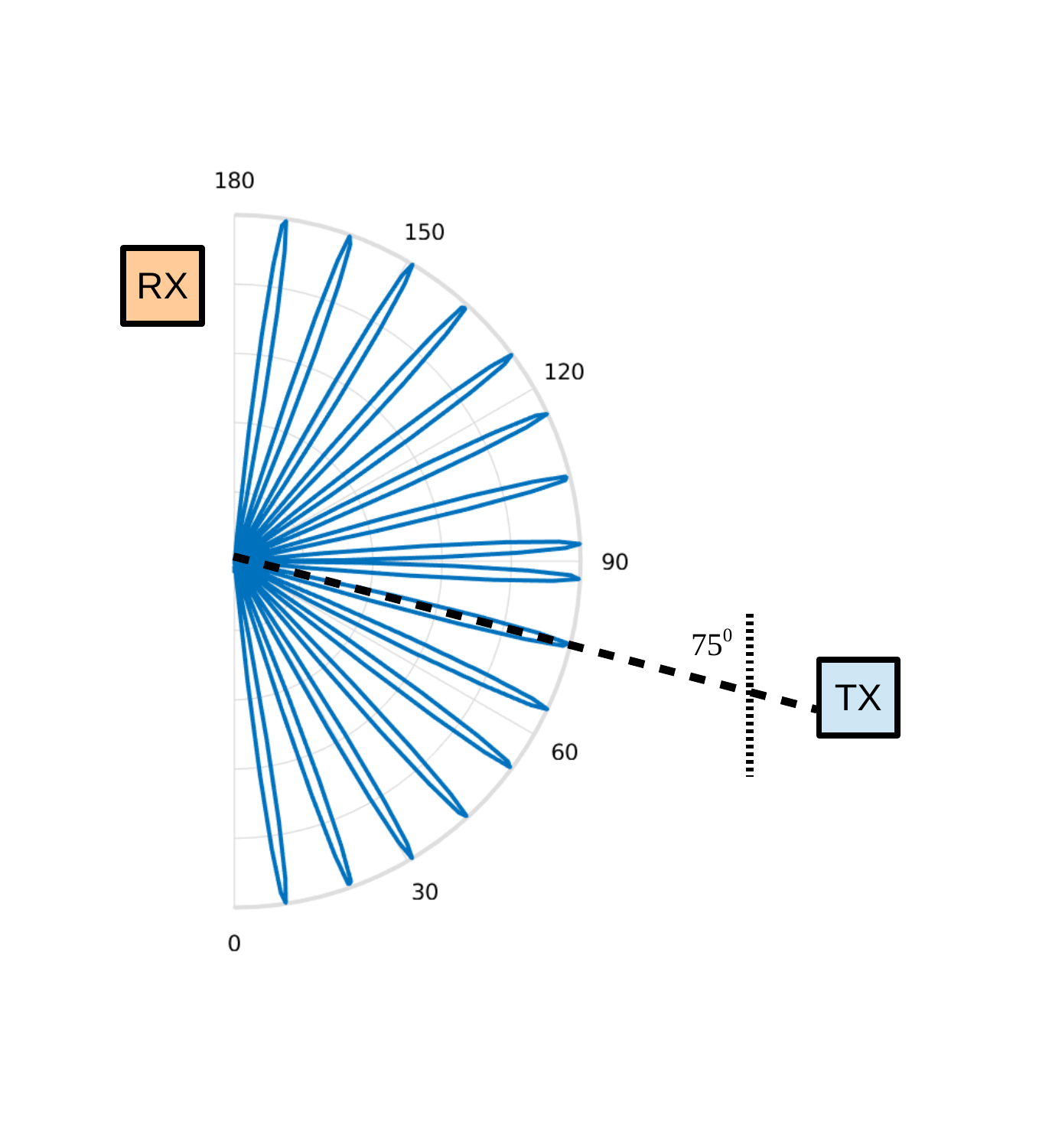}
\caption{{Beam alignment of the transmitter and receiver with 16 beams.} }
\label{fig:polar-example}
\end{figure}

\section{Related Work}

\subsection{Beam Alignment}
The authors in \cite{qi2016coordinated} perform initial access for clustered mmWave small cells using the power delay profile. In this case, base stations are coordinated in clusters, and communicate through a backhaul network. Base stations share their measurement reports obtained from the mobile devices, and location of the mobile is estimated based on the shared measurements. This will enable the base stations  to point at the estimated mobile location. Although this method is limited to line-of-sight scenarios, the probability of having at least three line-of-sight links (needed for mobile localization) increases by assuming larger cluster sizes.  In another line of research, the authors in \cite{desai2014initial} proposed a fast-discovery hierarchical search method, while \cite{singh2015feasibility} exploits the sparse multipath structure of the mmWave channel to optimize the choice of beamforming
directions. A cell discovery method is proposed
in \cite{barati2015directional} in which the base station periodically transmits synchronization signals to scan the entire angular space in time-varying random directions. 
In \cite{hur2013millimeter}, a beam alignment technique is designed based on adaptive subspace sampling and hierarchical beam codebooks. The authors in \cite{ali2017millimeter} use spatial information extracted at sub-6 GHz to help
estimate the best beam pairs at the transmitter and receiver at mmWave frequencies. In \cite{hassanieh2017agile}, a beam alignment scheme based on scanning several directions by one-shot is proposed.
In contrast to the previous works, we focus on exploiting contextual information in standalone mmWave systems in order to reduce the search space, and thus the overhead of beam alignment operation.  Note that there are other related work to reduce the overhead of beam search in integrated sub-6 GHz-mmWave systems \cite{hashemi2017energy,hashemi2017out}.

\subsection{Multi-Armed Bandit}
Multi-Armed Bandit (MAB) framework formulates sequential decision problems where an agent (i.e., decision maker) has to strike an optimal trade-off between exploitation and exploration by sequentially selecting
an action (or an arm), and observing the corresponding reward. Rewards of a given arm are random variables with unknown distributions. The objective is to maximize the expected reward over a given time horizon by selecting the optimal arm at each time slot.  Most of the existing works focus on \emph{unstructured MAB} problems in which the reward associated with different arms are not related \cite{robbins1985some,lai1985asymptotically}. In contrast to the unstructured MAB models, when the average rewards are structured, deriving the optimal regret bound and designing of optimal decision algorithms is more challenging \cite{bubeck2012regret}. Unimodal bandits are specific instances of bandit models in which the average reward of arms are correlated. In \cite{cope2009regret}, unimodal bandits with a continuous set of arms are studied, and the authors show that the regret of the order of $O(\sqrt{T})$ is achievable under some strong regularity assumptions on the reward functions. 
For the same problem, the authors in \cite{jia2011unimodal} provide an algorithm that achieves $O(\sqrt T \log(T))$ regret under relaxed regularity assumptions. In this paper, we cast the problem of mmWave beam alignment as a unimodal bandit model, and derive the regret bound. 

\section{Model and Objective} 
\subsection{System Model}
In mmWave systems, once the connection is lost, there are two options for connection establishment and subsequent beamforming: digital or analog. 
Digital beamforming is highly efficient in delay such that with the observations from all receive antennas, beamforming can be done by one-shot processing of the observed beacons. However, to achieve digital beamforming, there is the need for a separate analog-to-digital converter (ADC) for each antenna, which may not even be feasible for even a small to mid-sized antenna array due to high energy consumption. On the other hand, while analog beamforming requires only one ADC, it can focus on one direction at a time, making the search process costly in delay. 
Given the fragility of the mmWave channel, the need to scan the entire space leads to the loss of opportunities to utilize the mmWave channel upon its availability.  In order to avoid high energy consumption by mmWave components, we focus on analog beamforming in which a single RF chain is deployed at the transmitter and receiver.  Other implementations (e.g., hybrid architectures) may look at combinations of directions, which is out of scope of this work.
 \begin{figure}[t]
\hspace{-.5cm}
\includegraphics[scale=.3, trim = 1.6cm 1.2cm 2cm 1cm, clip]{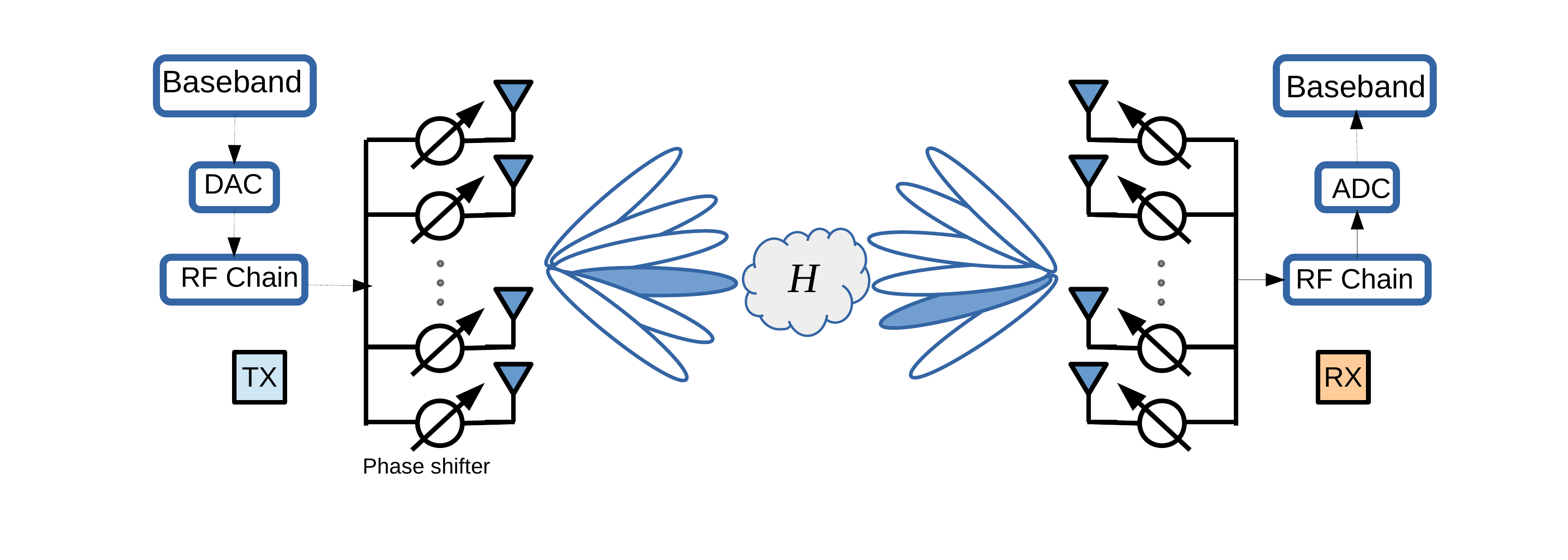}
\caption{Beam alignment using analog beamforming }
\label{fig:mmWave-model}
\end{figure}
Figure \ref{fig:mmWave-model} depicts a schematic of the analog beamforming, and directional beams at the transmitter and receiver. We assume that the transmitter and receiver are equipped with phased array antennas with $M_t$ and $M_r$ identical antennas respectively, equally spaced by a distance $d$ along an axis. For the sake of exposition, we only consider the receiver side, while a similar argument is held for the transmitter. 
%
%
Due to the use of analog architecture, the signal at the input of the decoder is a scalar, identical to a weighted combination of signal $x_{\text{m}}$ across all antennas. 
Thus, the received signal at the mmWave receiver can be written as: 
\begin{equation}
 y_{\text{m}}=\vc{w}_r^\dagger \mathbf{H} \vc{w}_t \cdot x_{\text{m}}+n_{\text{m}},
 \label{received-signal}
\end{equation}
where $\vc{w}_r$ and $\vc{w}_t$ are the beamforming vectors.  
The white Gaussian noise $n_{\text{m}}$ is normalized to have unit variance. 
 If the transmitter uses $N_t$ training precoding vectors $\vc{w}_t$, and the receiver uses $N_r$ training
combining vectors $\vc{w}_r$, then the collected signals (divided through by the training signal) is given by:
$$
\mathbf{Y} = \mathbf{W}_r^\dagger \mathbf{H} \mathbf{W}_t + \mathbf{N},
$$
where $\mathbf{W}_r=[\vc{w}_{r_1}, ..., \vc{w}_{r_{N_r}}]$ is the $M_r\times N_r$ combining matrix, and $\mathbf{W}_t=[\vc{w}_{t_1}, ..., \vc{w}_{t_{N_t}}]$ is the $M_t\times N_t$ precoding matrix. Furthermore, $\mathbf{N}$ is the $N_r \times N_t$ post-processing noise matrix. Hence, at each time slot, the problem of finding the best beam pair boils down to finding the largest value of matrix $\mathbf{Y}$. In the exhaustive search, one should examine all $N_r\times N_t$ elements of  $\mathbf{Y}$ to find the largest index, which determines the optimal beam index at the transmitter and receiver. The authors in \cite{ali2017millimeter} use spatial information extracted at sub-6 GHz to help
estimate the largest index of $\mathbf{Y}$. Applying the same framework, our beam alignment method exploits correlation across the elements of $\mathbf{Y}$ in order to reduce the search space to sub-matrices of $\mathbf{Y}$. 

\subsection{Problem Statement}
In order to establish a mmWave link, the transmitter selects a beam direction that determines the phase shifter weights to steer the beam in a certain direction. Similarly, the receiver selects a receive beam index to receive the signals in a certain direction. To obtain a high beamforming gain, the transmitter and receiver beams should be well aligned with each other. We let $D_{i,j} = (i, j)$ to denote a pair of beam direction in which $i$ is the beam index at the transmitter, and $j$ denotes the receiver beam index. There are $N_t$ and $N_r$ beams at the transmitter and receiver, respectively. Further, we define $\mathcal{D} = \{ (i, j) : 1 \leq i \leq N_t, 1 \leq j \leq N_r\} $ as the set of all possible beam pairs such that there exists $K = N_t \times N_r$ distinct \emph{matching} between the transmit and receive beams. For each pair of transmit and receive beams, \emph{misalignment} is defined as follows. 
  
\begin{definition}\textbf{(Misalignment)}
Given the pair $(i, j) \in \mathcal{D}$ of transmitter and receiver beams, the misalignment $\delta_{i,j}$ captures the angular mismatch between the $i$-th transmitter beam and  $j$-th receiver beam.
\end{definition}
\noindent Set $\mathcal{A} = \{\delta_{i,j}: 1 \leq i \leq N_t, 1 \leq j \leq N_r \}$ contains all possible values of the misalignment values such that $\mathcal{A}$ is partially ordered. With an abuse of notation, we let $\delta_k$ to denote the $k^{th}$ ($k\leq K$) misalignment value. 

\begin{figure}[t]
\includegraphics[scale=.5, trim=2cm 3.1cm 2cm 3.1cm, clip]{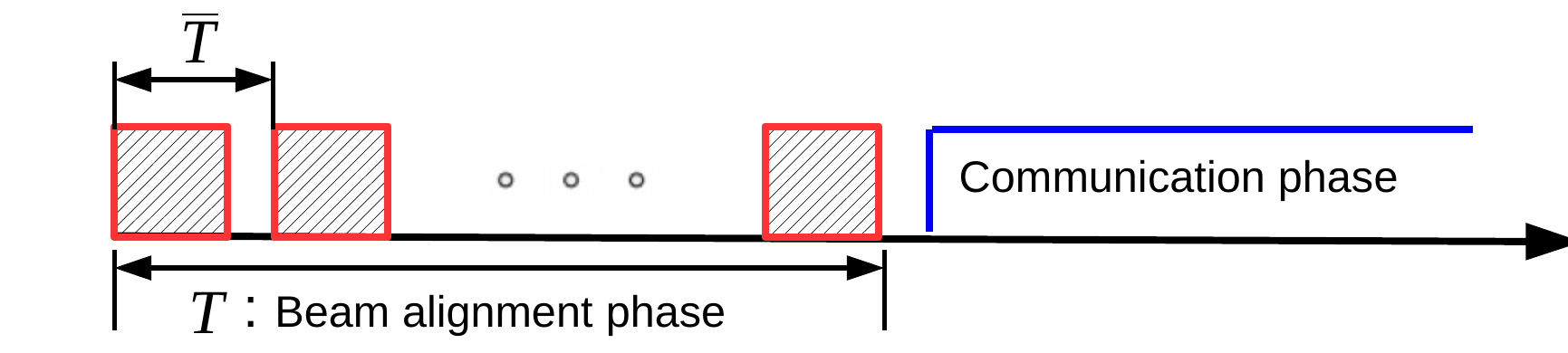}
\caption{{Sequence of beam alignment operations followed by a data communication phase.}}
\label{fig:time-slots}
\end{figure}

In order to detect the transmitted beacon signals, the received energy level should lie above a certain threshold $\tau$, which is determined based on the quality of service (QoS) needed. Once the received energy is larger than $\tau$, we call it a \emph{successful matching.}    For a fixed transmit power, the received signal energy can be expressed as a function of the misalignment value, i.e., $p(\delta_k)$ such that $p(.)$ is non-increasing. For the sake of notations, we use $p_k := p(\delta_k)$. The probability of success for the misalignment $\delta_k$ is given by $\theta_k = \mathds{P}(p_k \geq \tau)$. We assume a fixed situation (e.g., LOS) such that from one run of the alignment algorithm to another,  the situation remains fixed and only the orientations and distances change, i.e., the success probabilities $\theta_k$'s are time-invariant.  Furthermore, we let a binary random variable $X_k$ represent the success ($X_k = 1$) or failure ($X_k = 0$) of the matching. The optimal matching is unique, i.e., there exists $k^*$ such that $p_{k^*} \theta_{k^*} \geq p_{k} \theta_{k}$, for all $k \neq k^*$. 

\textbf{Problem formulation:}{Time is slotted, and we let $T$ to denote the length of beam alignment phase followed by the data communication phase. Further, $\bar{T}$ denotes the length of pilot signal by which each beam is measured. This value captures the amount of time that it takes to examine a single pair of beams, as shown in Fig. \ref{fig:time-slots}.  We formulate the problem of finding the best beam pair as an online stochastic optimization problem such that an optimal beam selection policy $\pi$ maximizes the expected amount of energy received from beacon messages up to a certain finite time $T$.} In this case, we let $s_k^\pi(T)$ denote the number of times that the misalignment $k$ is selected under policy $\pi$ and within the time period $T$.  Therefore, the optimal beam selection policy solves the following optimization problem:
\begin{subequations}
\begin{align}
& \max_\pi \ \  \sum_{k} \mathds{E}[{s_k^\pi(T)}] p_k \theta_k \\
 & \text{s.t.} \ \  \bar{T} \sum_k s_k^\pi(T) \leq T \ \text{and} \  s_k^\pi(T) \in \mathds{N}.
\end{align}
\end{subequations} 
In this formulation, small misalignment and large probability of success is desirable. In addition, given that examining each beam matching takes $\bar{T}$ on average, the total number of beam examinations is upper-bounded by $\frac{T}{\bar{T}}$, which is reflected in the first constraint. The second constraint implies that the number of each matching examination should be an integer.  

\section{Equivalent Multi-Armed Bandit Model}
The beam alignment formulation implies an MAB model such that  each combination of the transmitter and receiver beam is considered as an arm, which leads to $K=N_t\times N_r$ total arms. In this work, we use the terms ``arm'' and ``beam direction'' interchangeably. 
 In this case, $s_k^\pi(T)$ denotes the number of times that arm $k$ has been selected under policy $\pi$. Moreover, the reward of arm $k$ has Bernoulli distribution with parameter $\theta_k$ such that it is $p_k$ with probability $\theta_k$ and $0$ with probability $1-\theta_k$. The average reward of arm $k$ is denoted by $\mu_k := p_k \theta_k$.  

\subsection{Contextual Information} 
We explore a new type of contextual information that correlates the misalignment and the received energy. In particular, due to physics of signal propagation, if  matching at a larger misalignment is successful, a matching at a smaller misalignment will be successful with a high probability. On the other hand, if  matching at a smaller misalignment fails, then a matching at a larger misalignment will fail with a high probability as well. Hence, we have: 
\begin{equation}
\big( \delta_m > \delta_n \ \text{and} \ X_n  = 0 \big) \Rightarrow X_m = 0,
\end{equation}  
\begin{equation}
\big( \delta_m < \delta_n \ \text{and} \ X_n  = 1 \big) \Rightarrow X_m  = 1.
\end{equation} 
Equivalently, we can define the vector of success probabilities to satisfy the following condition: $\boldsymbol{\theta}  = (\theta_1, \theta_2, ..., \theta_K) \in \mathcal{T},
$
where 
$
 \mathcal{T}= \{ \boldsymbol{\theta} \in [0,1]^K: \theta_1 \geq \theta_2 \geq ... \geq \theta_K\}$. 
 In addition to the correlation property, we note that amount of energy received can be approximated as a unimodal function of misalignment. In this case, $\boldsymbol{\theta} \in \mathcal{U}$ such that $\mathcal{U} = \{\boldsymbol{\theta} \in [0,1]^K: \exists k^*, p_1 \theta_1 < ... < p_{k^*}\theta_{k^*},  p_{k^*}\theta_{k^*} > p_{k^*+1}\theta_{k^*+1} > ...> p_{K}\theta_{K} \}$. \blue{Note that a similar unimodal model has been used for other applications such as the rate adaptation in 802.11 systems \cite{combes2014optimal} and channel selection in cognitive networks \cite{combes2015dynamic}.}

\textbf{Graph representation:} In order to demonstrate the implications of contextual information, we can utilize a graph representation to capture the order of arms (beam pairs) with respect to each other. In this model, each arm corresponds to a node of a graph and each edge is associated with a relationship specifying which node of the edge gives the largest expected reward, thus providing a partial ordering over the arm space. Furthermore, from any node there is a
path leading to the unique node with the maximum expected
reward along which the expected reward is monotonically increasing. Under the assumption of unimodal expected reward, we can move from low expected rewards to high ones just by climbing them in the graph, preventing the need of a uniform exploration over all the graph nodes. This assumption reduces the complexity in the search for the optimal arm, since the optimal policy can avoid pulling the arms corresponding to some subset of non-optimal nodes.  
 
\textbf{Experimental observations:} {For the purpose of illustration, we provide experimental results to observe the typical propagation pattern as a function of misalignment. In particular, we consider the case of clear line of sight, and run a set of experiments in which two horn antennas placed on tripods and facing each other symmetrically. The transmitter antenna is set to be the stationary antenna, while the receiver antenna is rotated throughout the experiment. Two software defined radio (NI USRP-2901) are set up as the transmitter and receiver at $4$ GHz. The transmit antenna is connected to an up-converter with the output at $40$ GHz. The receiver antenna takes the $40$ GHz carrier and sends it to the down-converter with output of $4$ GHz. A power spectrum of gain (dB) vs. frequency is displayed in real time for data collection, and then we record the average peak value of gain. The receive horn antenna sweeps $2$ degrees incrementally on both clockwise and counter-clockwise directions until the gain is indistinguishable from the thermal noise floor (about $85$ dB).}
\begin{figure}[t!]
  \centering
  \subfigure[Two horn antennas at $40$ GHz placed on tripods]{\includegraphics[width=1.6in, height=2in, trim=1cm .1cm 1cm .75cm, clip]{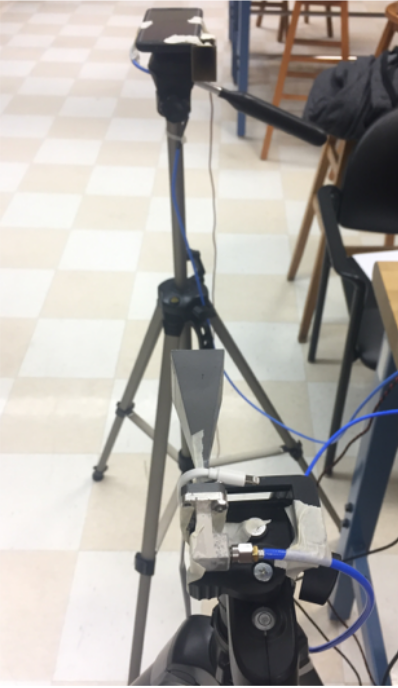}
  \label{antenna-pic}
  } \quad
  \subfigure[Connection of antenna with the oscillator, mixer and USRP]{\includegraphics[width=1.6in, height=2in, trim=1cm .1cm 1cm .75cm, clip]{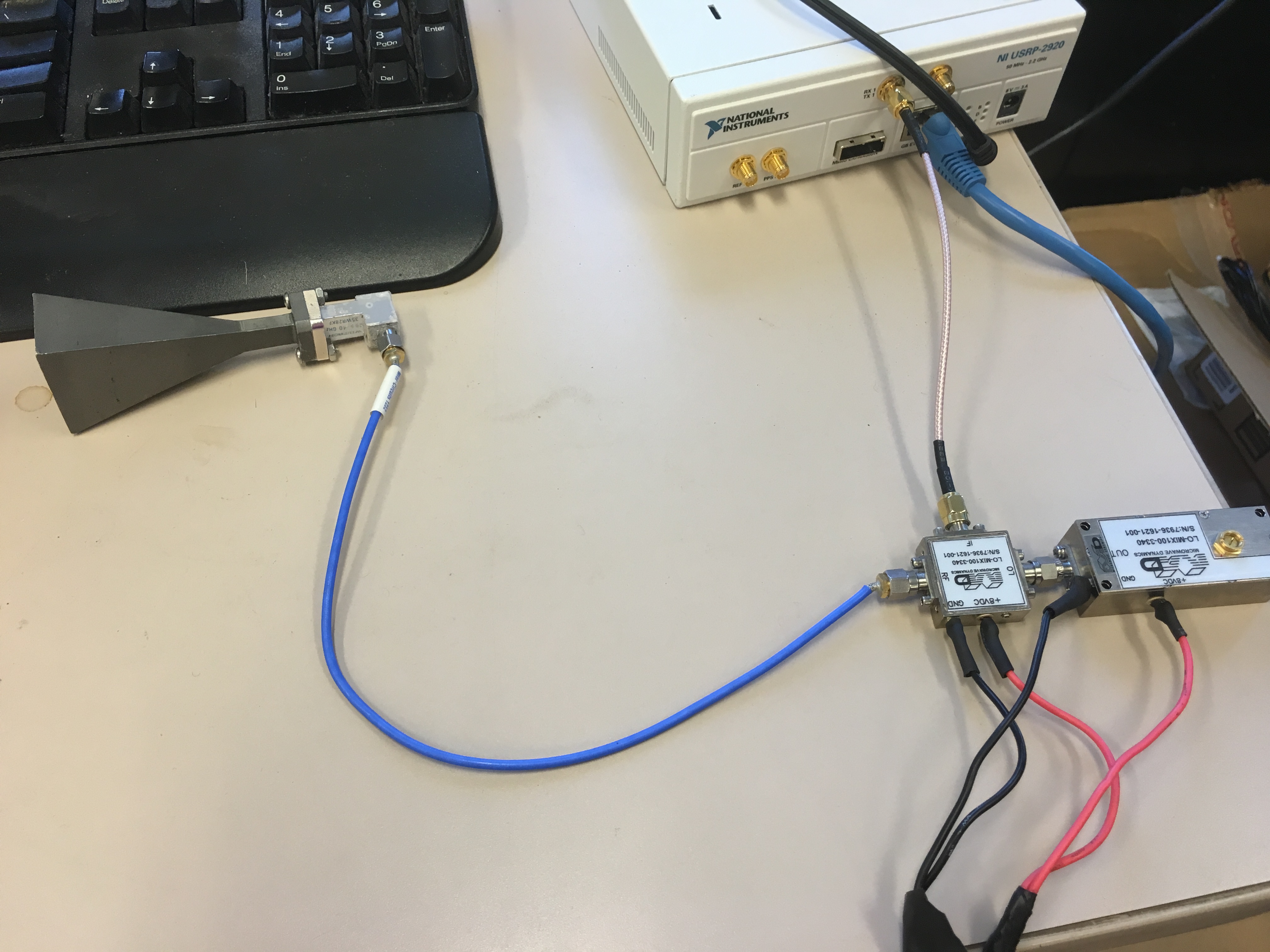}
  \label{connection-pic}
  }
 \caption[=]{Experimental setup connection and equipment}
  \label{fig:setup}
\end{figure}
\begin{figure}[t!]
\centering
\includegraphics[scale=.22, trim = 0cm 0cm 0cm 0cm clip]{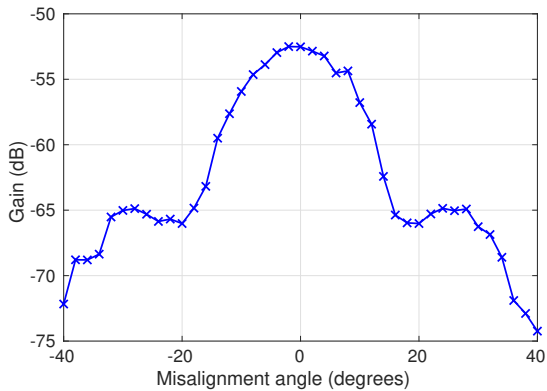}  
  \caption{Received power at distance of $1$ meter vs. misalignment. }
  \label{fig:experimental-unimodal}
\end{figure}
Our experimental setup is shown in Fig. \ref{fig:setup}, and Fig. \ref{fig:experimental-unimodal} demonstrates the received power as a function of misalignment angle between the transmitter and receiver antennas. 
We observe that received power approximately (due to the sidelobes) follows the unimodal pattern. In this work, we assume that the transmitter and receiver deploy highly directional antenna arrays in which the effect of sidelobes is negligible and thus the received power can be approximated as a unimodal function. 
Moreover, although we would expect the misalignment of $0$ degree provides the highest signal energy, under blockage and reflection scenarios a misalignment of $\delta \neq 0$ may provide a larger gain. Our model on contextual information is general and captures these conditions as well.




\section{Performance Analysis and Optimal Algorithm}

\subsection{Regret Analysis}
In order to assess the performance of policy $\pi$ for beam alignment, or equivalently arm selection, we consider \emph{regret} as the performance metric,  defined as follows:
\begin{equation}
\mathcal{R}^{\pi}(T) = p_{k^*} \theta_{k^*}  T - \sum_{k}  \mathds{E}[s_k^{\pi}(T)] p_k \theta_k.
\label{eq:regret-definition}
\end{equation}
From this definition, \emph{regret measures the expected reward loss (over a time period of $T$) compared with an oracle policy that would know everything.} In the case that the expected reward of the various arms are not correlated,  regret of the best algorithm is in the form of $O(K \log(T))$ in which $K$ is the number of arms \cite{lai1985asymptotically}. As a result, regret scales linearly with the number of arms. In beam alignment, for a beam of a few degrees the total number of arms (i.e., all combinations of transmitter and receiver beam pairs) becomes very large.  Therefore, in order to avoid the scaling factor, we exploit the structural properties of the arms and their reward functions, and show that due to contextual information, the scaling factor is constant and does not scale with the number of beam matchings (i.e., size of the decision space). To this end, for any arm $k$, we denote the set of its neighbors by:
$$
N(k) = \{j \in \{k-1, k+1\}: p_k \theta_k \leq p_j \}. 
$$
Furthermore,  given that the arms reward follow Bernoulli distribution, the Kullback-Leibler (KL) divergence of two Bernoulli distributions with respective parameters $\theta$ and $\theta^*$ is defined as: $I(\theta,\theta^*)=\theta\log \frac{\theta}{\theta^*} + (1-\theta) \log \frac{1-\theta}{1-\theta^*}$.  \blue{It has been shown in \cite{combes2014optimal,combes2014unimodal} that the problem of
learning in a unimodal bandit setting presents a lower bound over the regret of the following form:}
%

\begin{theorem}
For any beam alignment algorithm $\pi$, the lower bound on the regret is given by: 
\blue{
$$
\liminf_{T \rightarrow \infty} \frac{\mathcal{R}^{\pi} (T)}{\log(T)} \geq c(\theta),
$$}
in which $c(\theta)$ is a function of arms reward, and is given by:
\begin{equation}
c(\theta) = \sum_{k\in N(k^*)} \frac{p_{k^*}\theta_{k^*} - p_k \theta_k}{I\left(\theta_k, \frac{p_{k^*}\theta_{k^*}}{p_k}\right)}.
\label{eq:regert-scale}
\end{equation}
\end{theorem}


From \eqref{eq:regert-scale}, we observe that $c(\theta)$ is equal to a summation over a constant number of terms (i.e., independent of $K$). On the other hand, in the case that the structural properties is not exploited, the regret is lower bounded by:
$$
c'(\theta) = \sum_{k\neq k^*} \frac{p_{k^*}\theta_{k^*} - p_k \theta_k}{I\left(\theta_k, \frac{p_{k^*}\theta_{k^*}}{p_k}\right)},
$$ 
where $c'(\theta)$ linearly increases
with the number of possible arms, or equivalently, number of beam pairs at the transmitter and receiver. 


\subsection{Unimodal Beam Alignment (UBA) Algorithm} 
\blue{Next, we consider an algorithm whose regret
matches the lower bound given in Theorem 1. The first part of this algorithm is identical to the OSUB algorithm proposed in \cite{combes2014unimodal}, that we briefly describe here.} This  algorithm is asymptotically optimal, and is based on UCB algorithm that uses the KL divergence as an index for arm. In particular, each arm is attached an index that resembles the KL-UCB index, but the arm selected
at a given time is the arm with maximal index within the \emph{neighborhood} of the arm that yields the highest empirical reward. Let $k(t)$ be the arm selected at time $t$, and $s_k(t)$ denote the number of times arm $k$ has been selected up to time $t$. The empirical reward of arm $k$ at time $t$ is: 
\begin{equation}
 \hat{\mu}_k(t) = \left\{
  \begin{array}{l l}    
 \frac{\sum_{n=1}^t \mathds{1}\{k(n) = k\} p_k X_k(n)}{s_k(t)}  & \text{if} \ s_k(t) \neq 0 , \\ \vspace{.1cm}  
  
0 & \text{otherwise}. \vspace{.14cm} \\  

  \end{array} \right.
  \label{kkt-low-SNR}
\end{equation}
At any time slot $t$, we denote by $L(t) = \argmax_{1\leq k \leq K} \hat{\mu}_k(t)$ the index of the arm with the highest empirical reward. $L(t)$ is referred to as the \emph{leader} at time $t$.   Further, we define $l_k(t) = \sum_{n=1}^t \mathds{1}\{L(n) = k\}$ the number of times that arm $k$ has been the leader up to time $t$. Now, the index of arm $k$ at time $t$ is defined as:
\begin{align}
b_k(t) = \sup\bigg\{q\in [0, p_k] :  I \left(\frac{\hat{\mu}_k(t)}{p_k}, \frac{q}{p_k}\right) \leq \frac{f(t)}{s_k(t)}\bigg\},
\end{align}
in which $f(t) =\log (l_{L(t)}(t)) + c \log(\log(l_{L(t)}(t)))$ and $c$ is a positive number. At any time slot, the algorithm selects the arm ``close'' to arm $L(t)$ and with the maximum index. Next, we provide the finite time analysis of this algorithm, \blue{noting that the authors in \cite{combes2014optimal,combes2014unimodal} have presented similar results.}



\begin{theorem}
Let fix $\theta \in \mathcal{T} \cap \mathcal{U}$. For all $\epsilon $, the regret under the proposed UBA policy and at time $T$ is bounded by:
$$
\mathcal{R}(T) \leq (1+\epsilon) \sum_{k\in N(k^*)} \frac{p_{k^*}\theta_{k^*} - p_k\theta_k}{I\left(\theta_k, \theta_{k^*}\right)} \log(T). 
$$
\end{theorem}

\begin{proof}
Proof is provided in Appendix A. 
\end{proof}

\begin{algorithm}[t]
\caption{Unimodal Beam Alignment (UBA)}
\begin{varwidth}{\dimexpr\linewidth-2\fboxsep-2\fboxrule\relax}
\begin{algorithmic}[1]
\State At time slot $t \geq 1$, select the beam direction with index $k(t)$ where:
\begin{equation}
 k(t) = \left\{
  \begin{array}{l l}    
 L(t)  & \text{if} \ \frac{l_{L(t)}(t) -1 }{\gamma+1} \in \mathds{N} , \\ 
  
\displaystyle\argmax_{k \in N(L(t))} b_k(t) & \text{otherwise}.
  \end{array} \right.
\end{equation}
\blue{\State Evaluate the ratio 
$$
\psi(t) = \frac{p_{k(t)}}{\frac{1}{t}\sum_{n=1}^t p_{k(n)}},
$$
\If {$\psi(t) \geq \Psi$}
\State Terminate the beam search
\State Proceed to the communications phase with the beam index $k$ 
\EndIf} 
\end{algorithmic}
\end{varwidth}%
\end{algorithm}
 
In order to guarantee a finite time running of this algorithm, we add an additional termination condition in Algorithm 1 and continue with the data communication phase thereafter. We use the \textbf{peak-to-average ratio} as the termination condition in order to detect when the best beam direction is found.  Therefore, as the UBA algorithm proceeds, we evaluate the ratio of the received energy to the average of previously received signals energy. If the ratio is higher than a threshold $\Psi$, we terminate the UBA algorithm and declare the beam as the best beam direction. Specifically, at time slot $t$, we calculate the peak-to-average ratio as follows:
$$
\psi(t) = \frac{p_{k(t)}}{\frac{1}{t}\sum_{n=1}^t p_{k(n)}},
$$
in which $p_{k(t)}$ denotes the energy level of beam direction selected at time $t$. Therefore, when the condition $\psi(t)~\geq~\Psi$ is satisfied, we declare the beam with index $k(t)$ as the optimal direction, and the UBA algorithm stops at time $t$. 
The authors in \cite{nitsche2015steering} have experimentally evaluated the peak-to-average ratio for  LOS and NLOS situations such that $\Psi = 4$ is acceptable for detecting LOS. It should be noted that the proposed UBA scheme does not rely on the existence of LOS scenarios, while the threshold $\Psi$ can be different under various environmental conditions. In particular, environmental conditions (e.g., blockage or reflection) alter the success probability of beam matchings, while the proposed UBA is oblivious to the underlying ``physical layer" condition. In fact, based on the past observations, the UBA biases the search space towards the best beam direction.  Therefore, the transmitter and receiver are able to refine the search space through successive rounds of beam alignment. The pseudocode is provided in Algorithm 1 where $\gamma = 2$ is the maximum degree of the graph representing the relation between arms. 

In order to provide a complete beam alignment algorithm, similar to the IEEE 802.11ad standard, we decouple the transmitter and receiver steering such that the transmitter starts with a quasi-omnidirectional beam, while the receiver uses the UBA algorithm (instead of exhaustive search) to find the best beam direction. The process is then reversed to have the transmitter scan the space while keeping the receiver quasi-omnidirectional. As a result, we enhance the 802.11ad standard beam alignment by using the UBA algorithm instead of exhaustive search. 


%

\section{Numerical Results}

\subsection{Setup}
We compare the performance of the UBA algorithm with the exhaustive search scheme in which the receiver scans all different directions and samples the beam in all directions. The combination of transmitter and receiver beams that delivered the maximum power is picked as the direction of the signal. We perform the comparison of the UBA algorithm with the exhaustive search method under two different scenarios: \emph{directional} and \emph{quasi-directional}. Under directional conditions, probability of success is either very high or very low (e.g., in LOS scenarios).  On the other hand, quasi-directional scenario occurs when the variance of success probabilities is smaller than directional situation (e.g., NLOS conditions).  We evaluate the performance of beam alignment in terms of regret that measures the performance loss compared with the optimal alignment (i.e., an oracle policy). In this case, a lower regret implies a higher amount of received energy, and thus a higher accuracy in beam alignment. We also compare the beam alignment accuracy and delay overhead when the termination condition of peak-to-average ratio is used. In simulations, we fix the transmitter beam direction, and the receiver scans the angular domain to find the optimal beam direction.

\subsection{Regret Performance}
We set the vector of success probabilities as follows: 
$$
\boldsymbol{\theta}_{\text{Directional}}~=~(0.99, 0.98, 0.96, 0.93, 0.9, 0.1, 0.06, 0.04);
$$
$$\boldsymbol{\theta}_{\text{Quasi-directional}} = (0.95, 0.9, 0.8, 0.65, 0.45, 0.25, 0.15, 0.1).$$
Figure \ref{fig:regret} demonstrates the regret of the UBA method compared with the exhaustive beam sampling method under the directional and quasi-directional scenarios with $8$ beam directions. From the results, we observe that the regret increases over time since compared with an oracle policy, the \textbf{total performance loss} keeps increasing. However, the regret curve is concave and its rate of increase, decreases with time (i.e., error decreases). In addition, exploiting the structural properties using the UBA algorithm greatly reduces the regret that is equivalent to a higher amount of received energy. This implies a higher beam alignment accuracy that is proportional with the received energy. In addition, both methods achieve a lower regret under the the directional scenario, as expected.

\subsection{Scaling with the Size of System}
Due to recent advances in antenna technologies, large directional antenna arrays with much smaller form factors can be deployed in relatively small chip areas. As a result, spatial resolution and number of the beams can be very large at the transmitter and receiver. Within this context, we investigate the effect of number of beam pairs on the performance of UBA. Figure \ref{fig:regret-5} demonstrates the regret  metric for $K=8$ and $K=16$ beam pairs.  From the results, we observe that the performance of UBA scheme does not degrade with the number of beams that is a function of the number of antennas at the transmitter and receiver. This is a crucial property in massive antenna systems. Similar to Fig. \ref{fig:regret}, UBA scheme achieves a better performance compared with the exhaustive beam sampling method for both $K=8$ and $K=16$ beam pairs.

\subsection{Beam Alignment Accuracy and Delay Overhead}
Next, we investigate the accuracy and delay overhead of the proposed UBA algorithm combined with the peak-to-average ratio termination condition. We set the number of beams to be equal to $8$ beams at the receiver, and the goal is to find the best beam  (e.g., misalignment angle of zero). Using the exhaustive search method, $8$ time slots is needed to examine all beams and pick the one with the highest received energy. 
This method is deterministic in a sense that the output is correct with the guaranteed delay overhead of $8$ slots. On the other hand, our method finds the optimal beam direction with a high probability while its delay overhead is smaller than the exhaustive method.
  \begin{figure}[t!]
\centering
\includegraphics[scale=.21, trim = 1cm 0cm 1cm 1.2cm, clip]{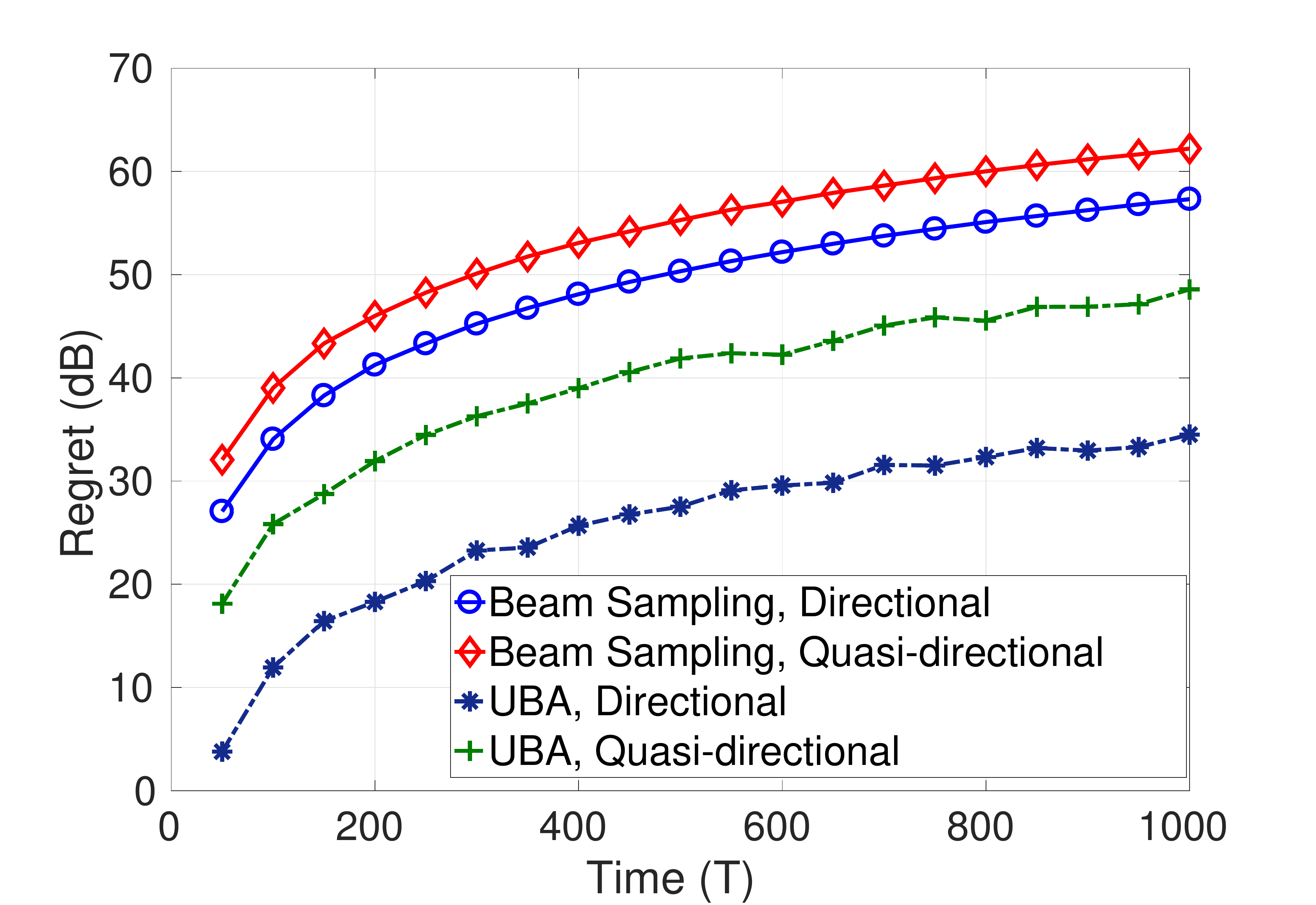}
\caption{{Regret of UBA and exhaustive beam search with $8$ beams.}}
\label{fig:regret}
\end{figure} 
  \begin{figure}[h!]
\centering
\includegraphics[scale=.21, trim = 1cm 0cm 1cm 1.3cm, clip]{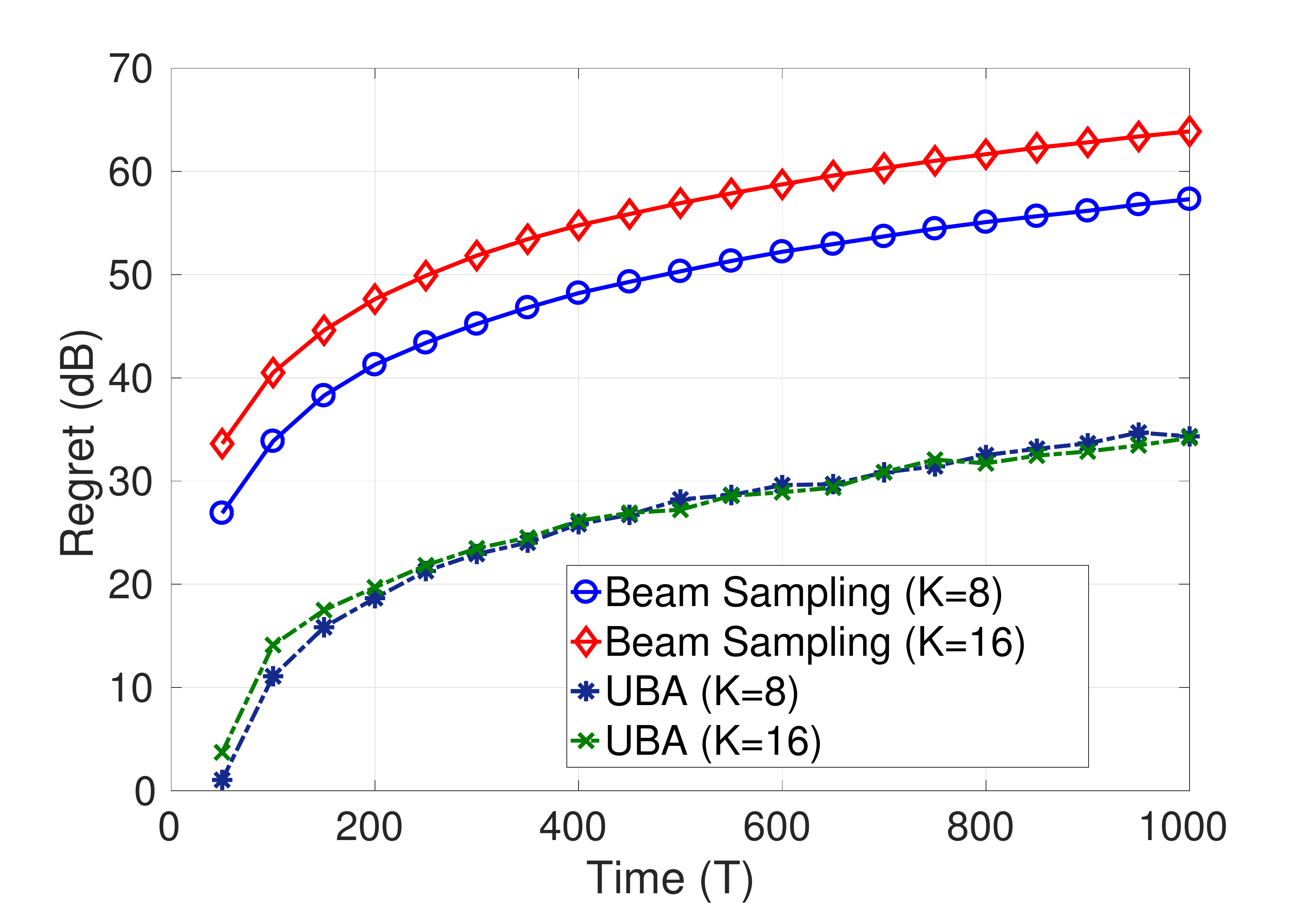}
\caption{{Regret of UBA and the exhaustive beam search with $8$ and $16$ beams. }}
\label{fig:regret-5}
\end{figure}
We set $\Psi = 4$ \cite{nitsche2015steering}, and consider a scenario in which beam alignment success probability for the optimal beam is relatively small, i.e, $\boldsymbol{\theta}~=~(0.8, 0.5, 0.35, 0.3, 0.25, 0.2, 0.15, 0.1)$. In this case, Fig. \ref{fig:accuracy} reports the CDF of the optimal beam detection. From the results, we observe that in more than $85\%$ of iterations, we correctly predict the optimal beam direction. 
The important point, however, is that our method significantly reduces the delay overhead. Figure \ref{fig:scatter} depicts the scatter plot for detecting each beam as the optimal vs. the amount of time it takes. Size of each scatter point represents density of data. From the results, we observe that most of the beam alignment operations lead to beam 1 (i.e., high accuracy) with delay of less than $5$ time slots (i.e., low overhead). Figure \ref{fig:delay-optimal-beam} also shows the CDF of delay overhead in detecting beam 1 as the optimal beam. We extend the simulation to $128$ receiver beams. From the results shown in Fig. \ref{fig:accuracy-delay-128}, we observe that delay overhead is significantly improved at the cost of some error in detecting the optimal beam direction. In fact, the delay overhead of $128$ time slots (to examine each direction) is reduced to $12.5$ time slots (averaged over 1000 iterations).

\section{Conclusion}
In this paper, we investigated the beam alignment problem in mmWave systems where the transmit and receive antenna arrays require 
to frequently find the optimal beam pair that maximizes the received energy from beacon messages.
   \begin{figure*}[t]
  \centering
\subfigure[UBA accuracy to detect beam index 1]{\includegraphics[scale=.22, trim=0cm .1cm 1cm .75cm, clip]{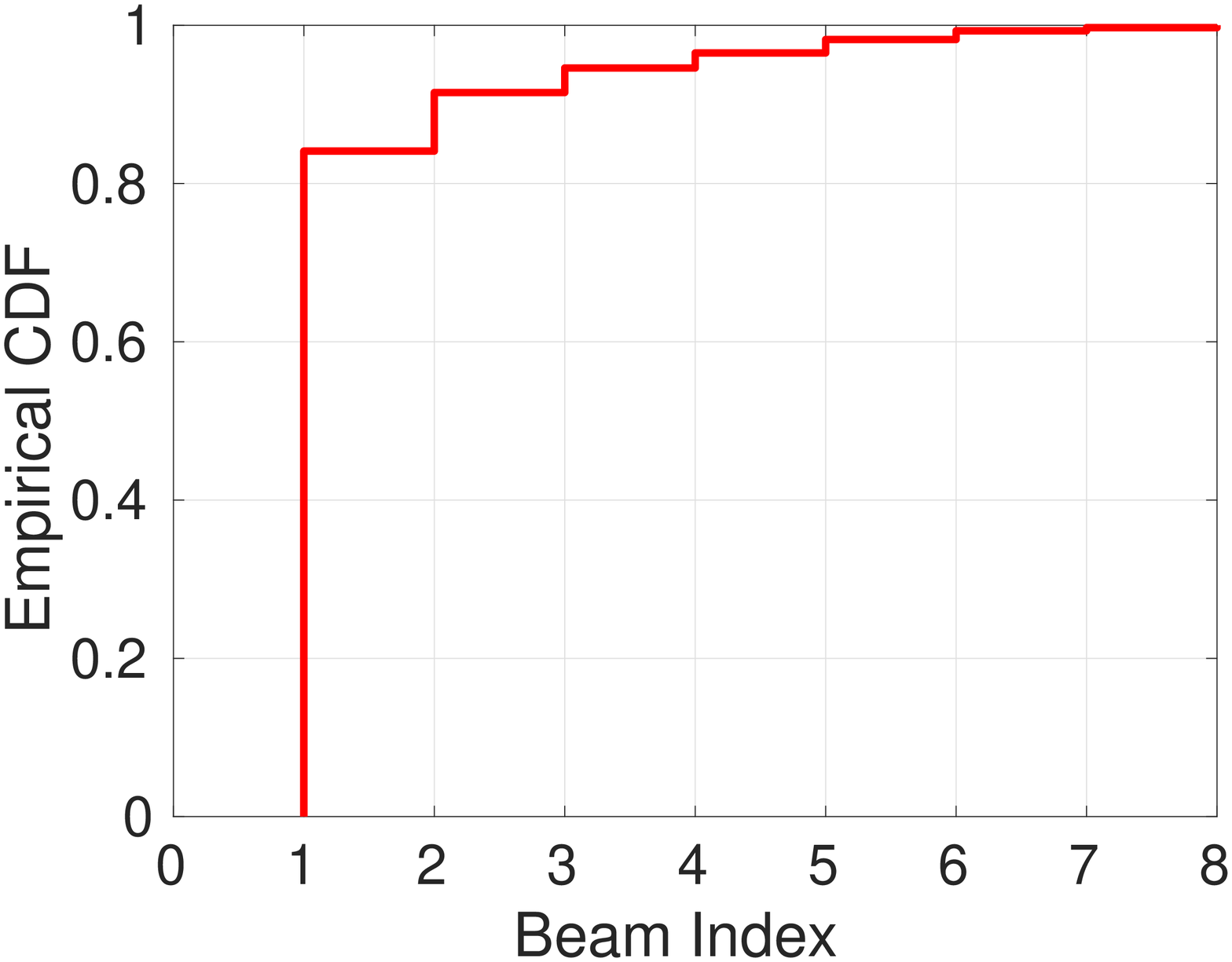}
  \label{fig:accuracy}
  } \quad
  \subfigure[Detected beam with its delay overhead]{\includegraphics[scale=.22, trim=.5cm .1cm 1cm .75cm, clip]{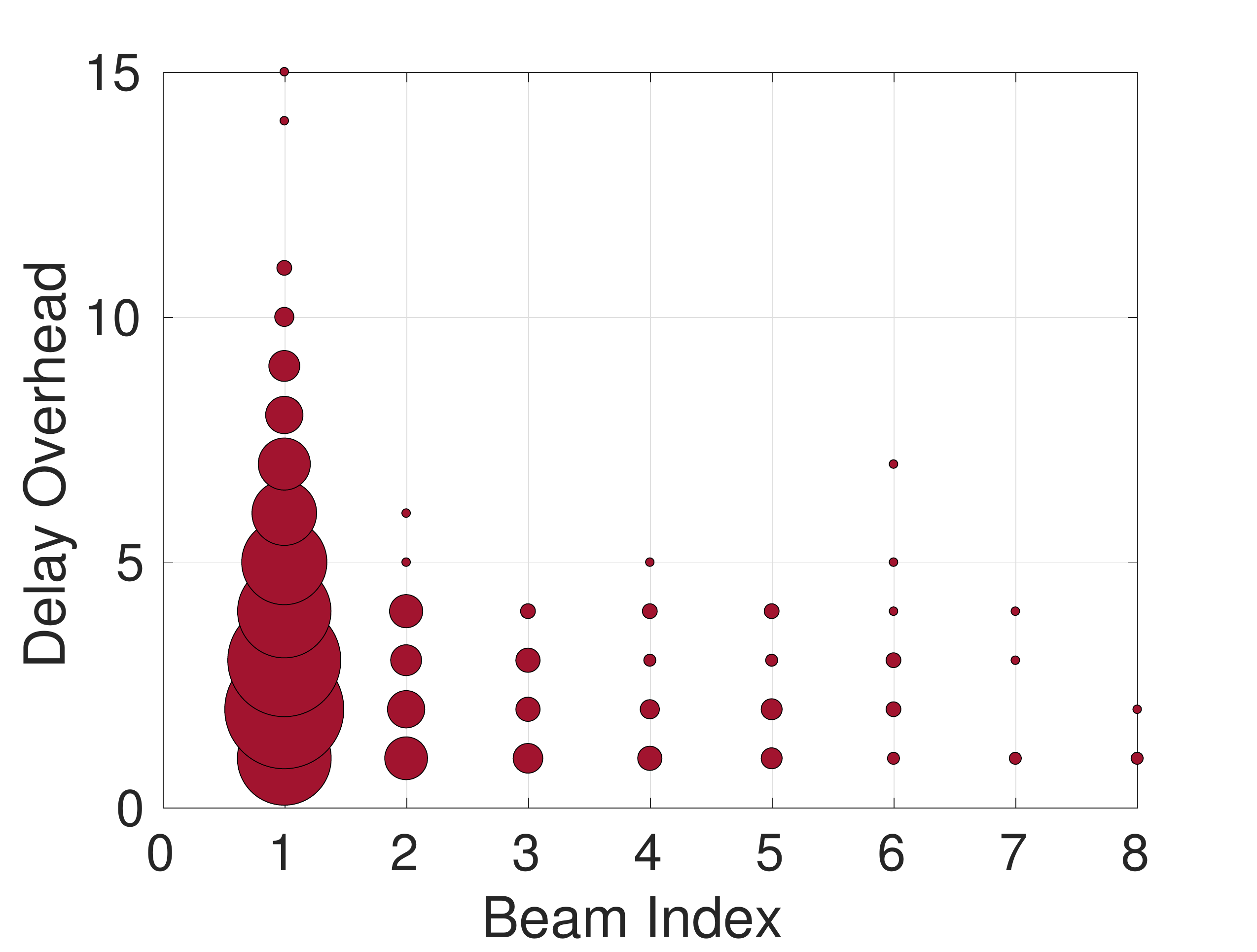}
  \label{fig:scatter}
  }
  \quad
  \subfigure[CDF of delay overhead to find the optimal beam]{\includegraphics[scale=.22, trim=0cm .1cm 1cm .75cm, clip]{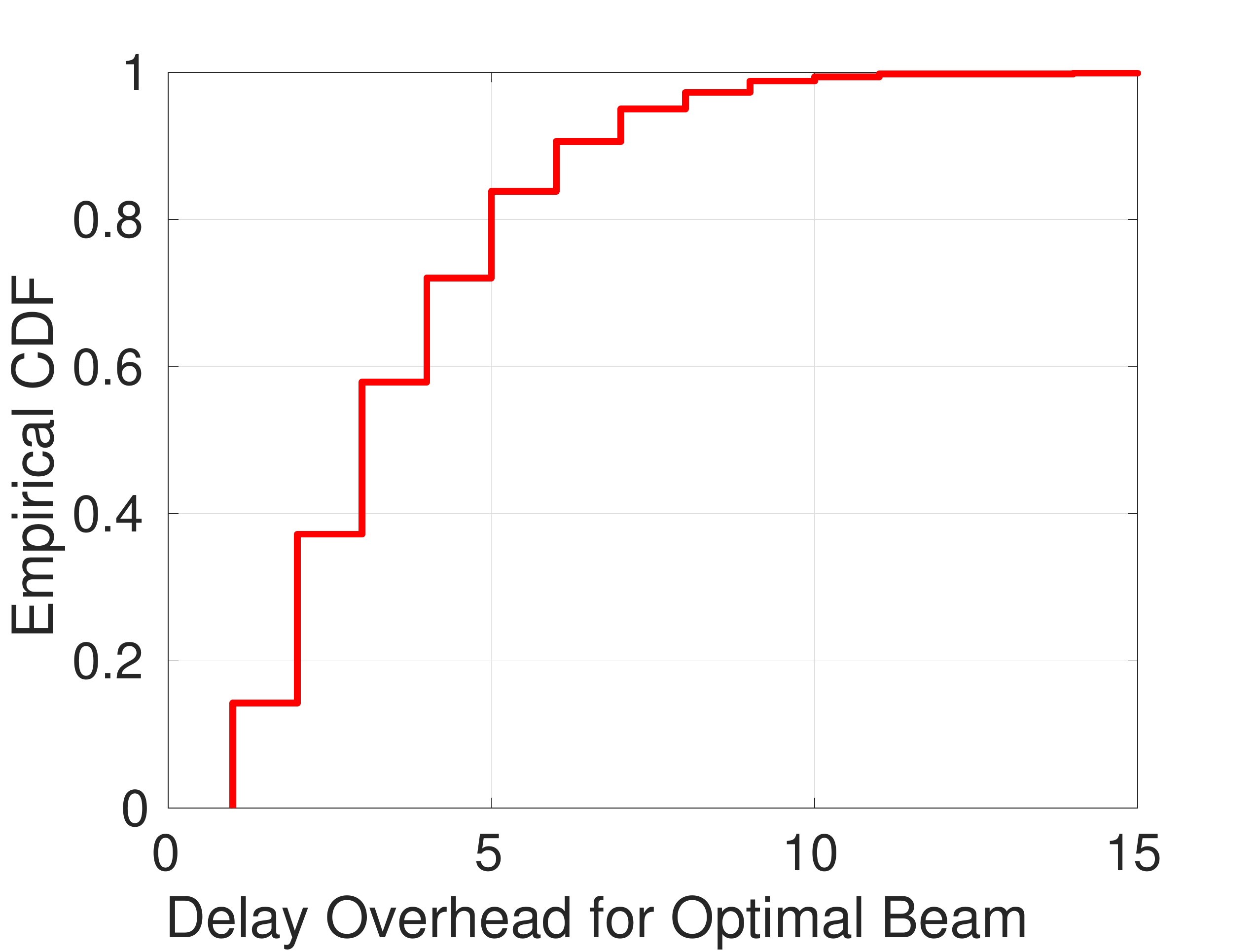}
  \label{fig:delay-optimal-beam}
  }
 \caption[=]{Accuracy and delay overhead of the proposed UBA algorithm with the detection threshold $\Psi = 4$ and \textbf{$\mathbf{8}$ beams.}}
  \label{fig:accuracy-delay}
\end{figure*} 
  \begin{figure*}[t]
  \centering
\subfigure[UBA accuracy to detect beam index 1]{\includegraphics[scale=.195, trim=.65cm .1cm 1cm .75cm, clip]{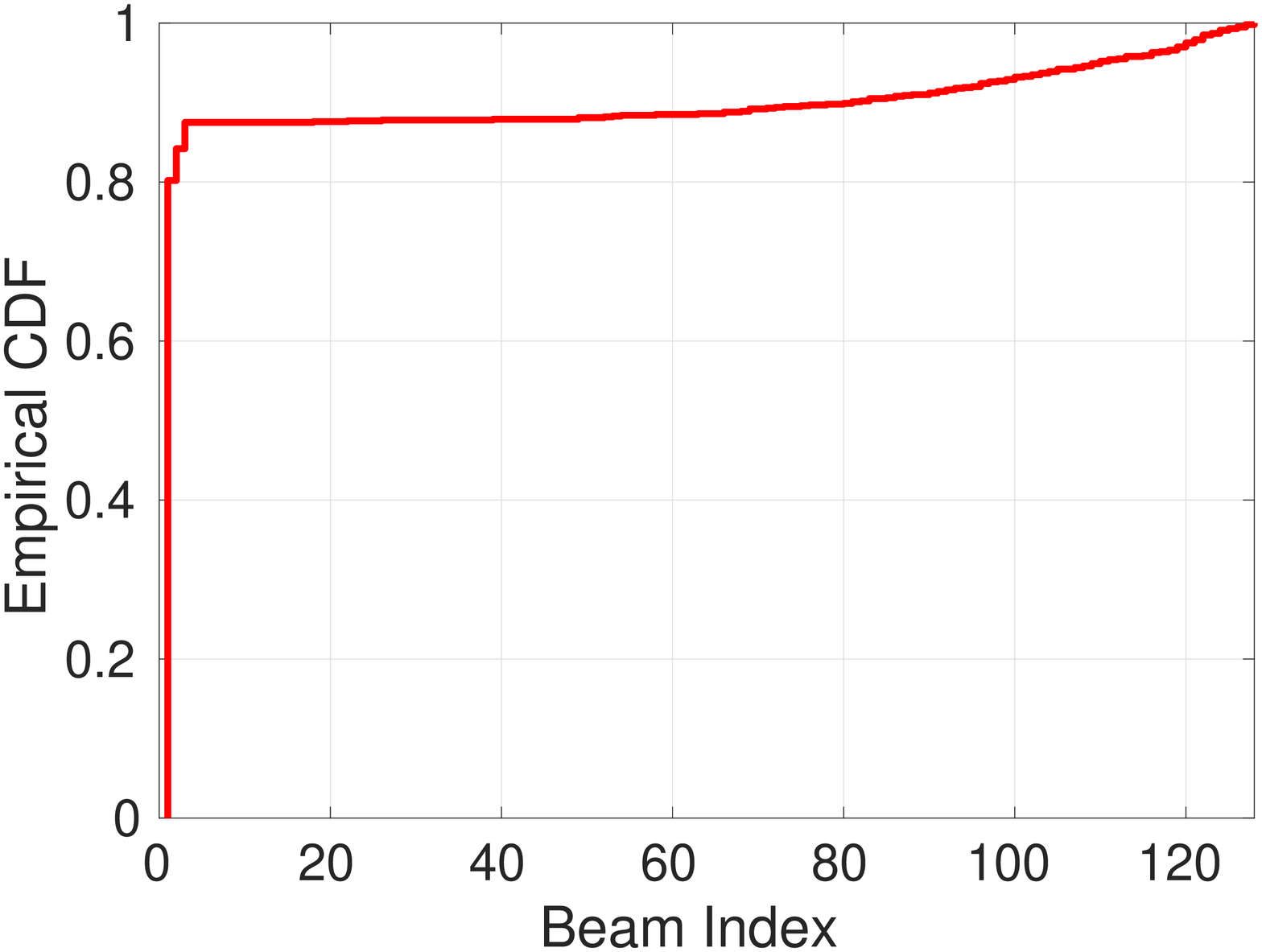}
  \label{fig:accuracy-128}
  } \quad
  \subfigure[Detected beam with its delay overhead]{\includegraphics[scale=.195, trim=.5cm .1cm 1cm .75cm, clip]{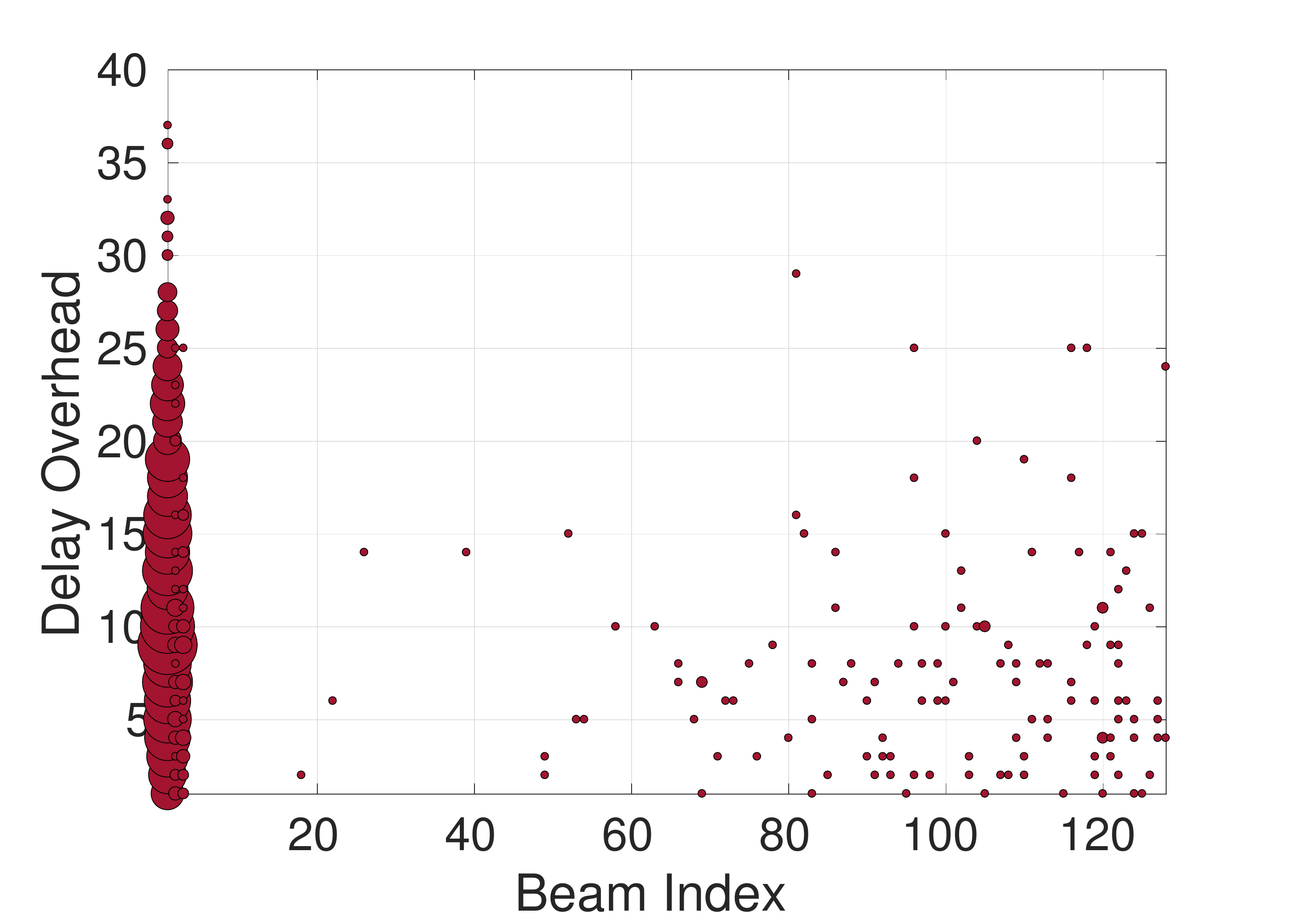}
  \label{fig:scatter-128}
  }
  \quad
  \subfigure[CDF of delay overhead to find the optimal beam]{\includegraphics[scale=.195, trim=0cm .1cm 1cm .75cm, clip]{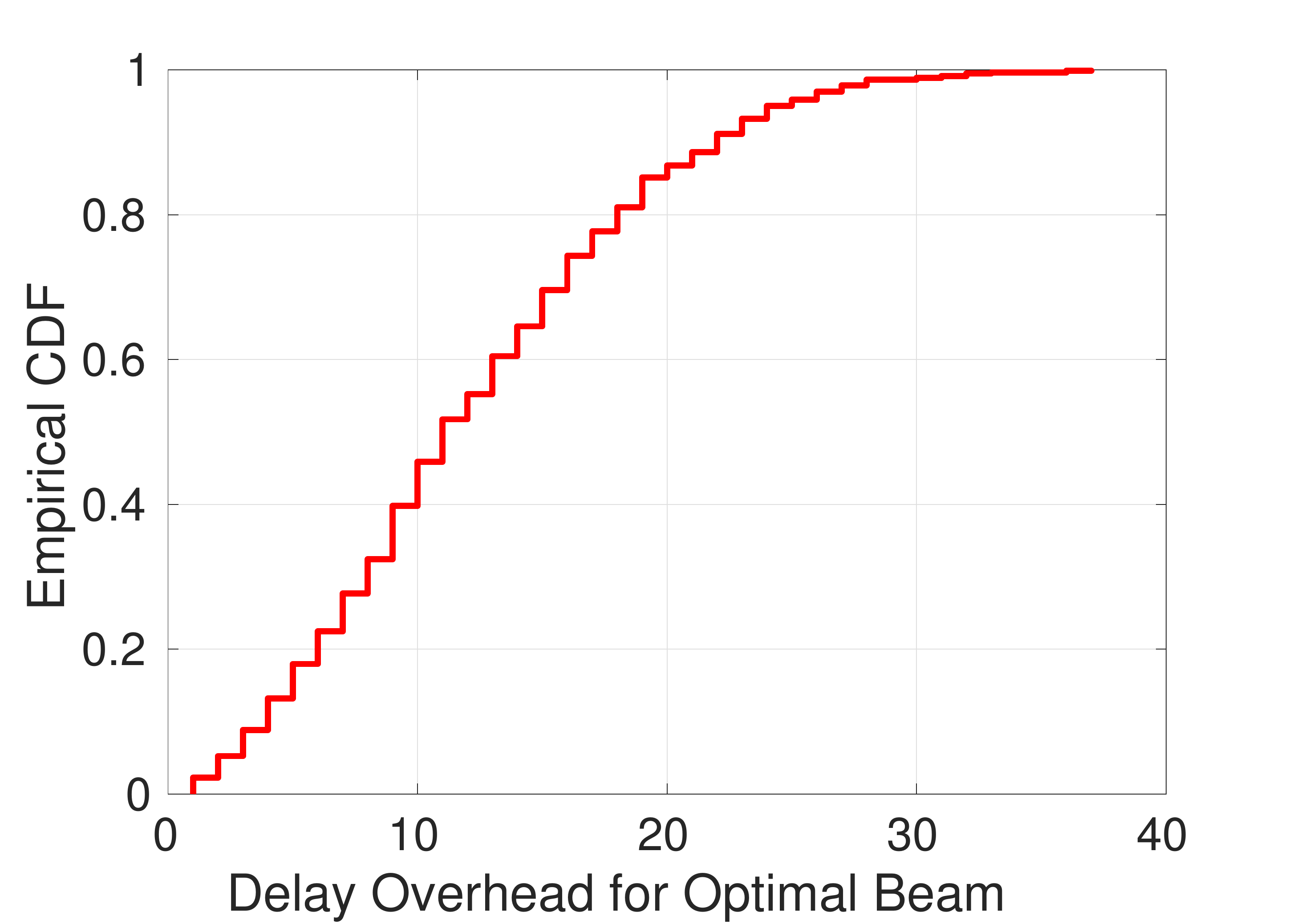}
  \label{fig:delay-optimal-beam-128}
  }
 \caption[=]{Accuracy and delay overhead of the proposed UBA algorithm with the detection threshold $\Psi = 4$ and \textbf{$\mathbf{128}$ beams.}}
  \label{fig:accuracy-delay-128}
\end{figure*}
 In order to reduce the overhead of exhaustive search methods, we investigated an online stochastic optimization problem and proposed an equivalent structured Multi-Armed bandit model. In this case, the problem of finding the best beam pair is reduced to finding the optimal arm at each time slot such that the overall regret is minimized. 
  We exploit the contextual information in order to reduce the search space, and thus the overhead of exhaustive beam selection.
  Thanks to the structural properties, we demonstrated that the regret bound does not depend on the size of decision space that is equal to the the number of transmit and receive beams multiplied. This is a crucial property in MIMO settings in which the number of all combinations of transmit and receive beams grows quickly. We further proposed an asymptotically optimal algorithm for the beam alignment problem and demonstrated its performance via simulations. 

 \section*{Acknowledgment}
This work was supported in part by the following NSF grants: CNS-1518916, CNS-1314822, CNS-1618566, CNS-1514260, CNS-1518829, and ONR grants: N00014-15-1-2166 and N00014-17-1-2417.  

The authors would like to thank Hongliang Si and Nathan Weirich for performing the experiments.

\appendices

\section{Proof of Theorem 2}
\begin{proof}
\blue{Similar to \cite{combes2014unimodal,paladino2017unimodal},}  we split the $T$ rounds in two sets: those rounds
in which the best arm $k^*$
is the leader , i.e., $L(t) = k^*$, and those in which the leader is another arm, i.e., $L(t) \neq k^*$. Therefore:

\small
\begin{align*}
\mathcal{R}(T) = \sum_{k \neq k^*} (\mu_{k^*} - \mu_k)  \mathds{E} \big[\sum_{t=1}^T & \mathbf{1}\{k(t) = k\}\big] & \nonumber \\ = \sum_{k \neq k^*} (\mu_{k^*} - \mu_k)  \mathds{E}  \big[\sum_{t=1}^T \mathbf{1}\{L(t) = k^* \ & \text{and} \ k(t)  = k\}\big]  \nonumber \\ + \sum_{k \neq k^*} (\mu_{k^*} - \mu_k) \mathds{E} \big[\sum_{t=1}^T \mathbf{1}\{L(t) \neq k^* & \ \text{and} \ k(t) = k\}\big].
\end{align*}
\normalsize

\noindent
If we consider the first term, the proposed algorithm behaves like the UCB algorithm restricted to the optimal arm and its neighborhood, and the regret upper bound is the one presented in \cite{cappe2013kullback}, i.e., for every $\epsilon > 0$: 

\small
$$
\mathcal{R}_1 (T) \leq (1+\epsilon) \sum_{k\in N(k^*)} \frac{\mu_{k^*} - \mu_k}{I (\theta_k, \theta_{k^*})}\big[ \log(T) + \log(\log(T))\big] + C, 
$$
\normalsize
 where $C$ is a constant. For the second part, we have:

\small
\begin{multline*}
\mathcal{R}_2 (T)= \sum_{k \neq k^*} (\mu_{k^*} - \mu_k) \mathds{E} \left[\sum_{t=1}^T \mathbf{1}\big\{L(t) \neq k^* \ \text{and} \ k(t) = k\big\}\right],
\end{multline*}
\normalsize

\noindent or $\mathcal{R}_2 (T) \leq \sum_{k\neq k^*}
 \mathds{E}[l_{k}(T)]$. Next, we provide an upper bound on the number of times that arm $k$ has been the leader, i.e., $l_{k} (T)$, with $\hat{l}_{k}(T)$ that is the number of rounds spent with arm $k$ as leader in the case only its neighborhood is considered during the whole time horizon $T$. Therefore, we have: 

\small
\begin{align}
\mathcal{R}_2 (T) \leq \sum_{k\neq k^*} \mathds{E}[l_{k}(T)] \leq \sum_{k\neq k^*}  \mathds{E}& [\hat{l}_{k}(T)] = \sum_{k \neq k^*} \sum_{t=1}^T \mathds{E} [\mathbf{1} \{L(t) = k \}]  \nonumber \\ = \sum_{k\neq k^*} \sum_{t=1}^T \mathds{E}\big[\mathbf{1}\{\hat{\mu}_{k}(t) = &\max_{j \in N(k)} \hat{\mu}_{j}(t)\}\big] ,
\end{align}
\normalsize

\noindent where, with an abuse of notations, $L(t)$ denotes the leader at round $t$ in this modified problem where only $N(k)$ is considered.  Since arm $k$ is the leader, its empirical mean is the maximum in its neighborhood, i.e., $\hat{\mu}_{k}(t)~\geq~\hat{\mu}_{k'} (t)$ in which $k' = \argmax_{i \in N(k)} \mu_i$. Thus, we have: 
$
\mathcal{R}_2 (T) \leq \sum_{k\neq k^*} \sum_{t=1}^T \mathds{E} [\mathbf{1}\{\hat{\mu}_{k}(t)~\geq~\hat{\mu}_{k'} (t)\}]  = \sum_{k\neq k^*} \sum_{t=1}^T \mathds{P} \big(\hat{\mu}_{k}(t) \geq \hat{\mu}_{k'} (t)\big). 
$
Defining $\Delta_k = \max_{k' \in N(k)} \hat{\mu}_{k'} - \mu_k$ as the expected loss incurred in choosing arm $k$ instead of its best adjacent one $k'$, we have:

\vspace{-.5cm}
\small
\begin{multline}
\vspace{-.3cm}
\mathcal{R}_2(T) \leq \sum_{k\neq k^*} \sum_{t=1}^T \mathds{P} \big(\hat{\mu}_{k}(t) - \mu_k - \frac{\Delta_k}{2} - \hat{\mu}_{k'}(t) + \mu_{k'} - \frac{\Delta_k}{2} \geq 0 \big) \nonumber \\
\leq \sum_{k\neq k^*} \bigg[\ \underbrace{\sum_{t=1}^T \mathds{P} \big(\hat{\mu}_{k}(t) - \mu_k - \frac{\Delta_k}{2} \geq 0 \big)}_{\mathcal{R}_{2,1}^k (T)}  + \nonumber \\ \underbrace{\sum_{t=1}^T \mathds{P} \big(\hat{\mu}_{k'} (t) - \mu_{k'} + \frac{\Delta_k}{2}\big) \leq 0}_{\mathcal{R}_{2,2}^k(T)} \bigg] 
\end{multline}
\normalsize

\vspace{-.4cm}
\noindent For the first term, we have: 
 
 \vspace{-.6cm}
 \small
\begin{align}
\mathcal{R}_{2,1}^k (T) =  \sum_{t=1}^T \mathds{P} (\hat{\mu}_{k}(t) \geq \mu_k + \frac{\Delta_k}{2})  \nonumber \\ 
= \sum_{t=1}^T  \sum_{h=1}^t \mathds{P}  (s_k(t) = h |  \hat{\mu}_k (t) \geq  \mu_k + \frac{\Delta_k}{2}) \mathds{P} (\hat{\mu} _k(t) & \geq \mu_k + \frac{\Delta_k}{2}) & \nonumber \\ 
\leq \sum_{t=1}^T \sum_{h=1}^t \mathds{P} (s_k (t) = h | \hat{\mu}_k (t) \geq \mu_k + \frac{\Delta_k}{2})  e^{-\frac{h\Delta_k^2}{2}},
\end{align}  
\normalsize

\noindent where the last inequality is due to the Chernoff-Hoeffding inequality expressed as follows:
\begin{lemma} (Chernoff-Hoeffding inequality) Let $X_1, ..., X_n$ be random variables with common
range $[0, 1]$ and such that $\mathds{E}[X_t|X_1, ...,  X_{t-1}] = \mu$. Let $S_n = X_1 +...+ X_n$. Then for all $a\geq 0$, we have: 
$\mathds{P} (S_n \geq n \mu+a) \leq e^{-2a^2/n}.$
\end{lemma}

\noindent Therefore, we have: 

\vspace{-.4cm}
\small
\begin{multline}
\mathcal{R}_{2,1} ^{k} (T) \leq \sum_{t=1}^T \bigg( \sum_{h=1}^{x_0} \mathds{P} (s_k (t) = h | \hat{\mu}_k (t) \geq \mu_k + \frac{\Delta_k}{2})e^{-\frac{h\Delta_k^2}{2}} \nonumber \\ + \frac{2}{\Delta_k^2} e^{\frac{-x_0 \Delta_k^2}{2}}\bigg), 
\end{multline}
\normalsize 

\noindent since $\sum_{t=x+1}^\infty e^{-kt} \leq \frac{1}{k} e^{-kx}$. Then, we have: 
$
\mathcal{R}_{2,1} ^k (T) \leq \sum_{t=1}^T \bigg( x_0 \mathds{P} (s_k (t) \leq x_0)e^{-\frac{h\Delta_k^2}{2}} + \frac{2}{\Delta_k^2} e^{\frac{-x_0 \Delta_k^2}{2}}\bigg), 
$
where we dropped the conditioning. Since the expected number of times that the non-optimal arm has been played is bounded and its variance is bounded as well, using Bernstein’s inequality (provided below), we have: 
$
\mathds{P} \big(s_k (t) \leq x_0\big) \leq e^{-x_0/5},
$
and since $x_0$ is lower bounded, we conclude that $\mathcal{R}_{2,k} (T)$ is bounded by a constant, i.e., $\mathcal{R}_{2,1}^k(T) \leq C'$.
\begin{lemma} (Bernstein inequality)
Let $X_1, ..., X_n$ be random variables with range in $[0,1]$ and 
$
\sum_{t=1}^n \text{Var}[X_t|X_{t-1}, ..., X_1] = \sigma^2.
$
Let $S_n = X_1 + X_2 + .. + X_n$. Then for all $a \leq 0 $, we have: 
$
\mathds{P}\big(S_n - \mathds{E}[S_n] \geq a\big) \leq e^{\frac{-a^2/2}{\sigma^2+a/2}}.
$
\end{lemma}

\noindent For the $\mathcal{R}_{2,2}^k (T)$ term, we have: 

\vspace{-.4cm}
\small
\begin{multline}
\mathcal{R}_{2,2}^k (T) = \sum_{t=1}^T \mathds{P} (\hat{\mu}_{k'} (t) \leq \mu_{k'} - \frac{\Delta_k}{2}) \nonumber \\ 
= \sum_{t=1}^T \sum_{h=1}^t \mathds{P} (s_{k'}(t) = h \ \text{and} \ \hat{\mu}_{k'} (t) \leq \mu_k' - \frac{\Delta_k}{2}) \nonumber \\ 
= \sum_{t=1}^T \sum_{h=1}^{t_b} \mathds{P} (s_{k'}(t) = h \ \text{and} \ \hat{\mu}_{k'} (t) \leq \mu_k' - \frac{\Delta_k}{2}) \nonumber \\ 
+ \sum_{t=1}^T \sum_{h=t_b+1}^t \mathds{P} (s_{k'}(t) = h | \hat{\mu}_{k'} (t) \leq \mu_k' - \frac{\Delta_k}{2}) \mathds{P} (\hat{\mu}_{k'} (t) \leq \mu_k' - \frac{\Delta_k}{2}) \nonumber \\
\leq \sum_{t=1}^\infty \mathds{E} [\mathbf{1}\{s_{k'}(t) \leq t_b\}] + \sum_{t=1}^T \sum_{h=t_b+1}^t \mathds{P} (\hat{\mu}_{k'} (t) \leq \mu_k' - \frac{\Delta_k}{2}) \nonumber \\
\leq C + \sum_{t=1}^T \sum_{h=t_b+1}^t e^{-\frac{t\Delta_k^2}{2}}.
\end{multline}
\normalsize

\noindent It is straightforward to show that $\mathcal{R}_{2,2}^k (T)$ is bounded by a constant as well. By considering the three partial results on $\mathcal{R}_1$, $\mathcal{R}_{2,1}^k$, and $\mathcal{R}_{2,2}^k$, we have:
$$
\mathcal{R} (T) \leq \mathcal{R}_1 (T) + \sum_{k \neq k^*} \mathcal{R}_{2,1}^k (T) + \mathcal{R}_{2,2}^k (T),
$$
and the theorem statement follows. 
\end{proof}

\small
\bibliographystyle{IEEEtran}
\bibliography{../References}
\end{document}